\def\year{2022}\relax
\documentclass[letterpaper]{article} 
\usepackage{aaai22}  
\usepackage{times}  
\usepackage{helvet}  
\usepackage{courier}  
\usepackage[hyphens]{url}  
\usepackage{graphicx} 
\usepackage{pgfplots}
\usepackage{natbib}  
\usepackage{caption} 
\DeclareCaptionStyle{ruled}{labelfont=normalfont,labelsep=colon,strut=off} 
\usepackage{algorithm}
\usepackage{algorithmic}
\usepackage{amsthm}
\usepackage{amsmath}
\usepackage{amssymb}
\usepackage{cleveref}
\usepackage{subfigure}
\usepackage{siunitx}

\usepackage{lineno}
\newcommand*\patchAmsMathEnvironmentForLineno[1]{
  \expandafter\let\csname old#1\expandafter\endcsname\csname #1\endcsname
  \expandafter\let\csname oldend#1\expandafter\endcsname\csname end#1\endcsname
  \renewenvironment{#1}
     {\linenomath\csname old#1\endcsname}
     {\csname oldend#1\endcsname\endlinenomath}}
\newcommand*\patchBothAmsMathEnvironmentsForLineno[1]{
  \patchAmsMathEnvironmentForLineno{#1}
  \patchAmsMathEnvironmentForLineno{#1*}}
\AtBeginDocument{
\patchBothAmsMathEnvironmentsForLineno{equation}
\patchBothAmsMathEnvironmentsForLineno{align}
\patchBothAmsMathEnvironmentsForLineno{flalign}
\patchBothAmsMathEnvironmentsForLineno{alignat}
\patchBothAmsMathEnvironmentsForLineno{gather}
\patchBothAmsMathEnvironmentsForLineno{multline}
}

\newtheorem{theorem}{Theorem}
\newtheorem{lemma}{Lemma}

\newtheorem{observation}{Observation}
\newtheorem{definition}{Definition}

\newcommand{\dist}{{d}}

\usepackage{newfloat}
\usepackage{listings}
\floatstyle{ruled}
\newfloat{listing}{tb}{lst}{}
\floatname{listing}{Listing}
\nocopyright

\title{Computing Diverse Shortest Paths Efficiently:\\ A Theoretical and Experimental Study}
\author{
     Tesshu Hanaka\textsuperscript{\rm 1}, Yasuaki Kobayashi\textsuperscript{\rm 2}, Kazuhiro Kurita\textsuperscript{\rm 3}, See Woo Lee\textsuperscript{\rm 2}, Yota Otachi\textsuperscript{\rm 1}\\
}
\affiliations{
    \textsuperscript{\rm 1} Nagoya University, Nagoya, Japan Email: $\{$hanaka, otachi$\}$@nagoya-u.jp\\
    \textsuperscript{\rm 2} Kyoto University, Kyoto, Japan Email: $\{$kobayashi, slee06$\}$@iip.ist.i.kyoto-u.ac.jp\\
    \textsuperscript{\rm 3} National Institute of Informatics, Tokyo, Japan Email: kurita@nii.ac.jp\\
}

\usepackage{bibentry}

\begin{document}

\maketitle

\begin{abstract}
    Finding diverse solutions in combinatorial problems recently has received considerable attention~\cite{Baste:Diversity:2020,Fomin:Diverse:2020,Hanaka:Finding:2021}. 
    In this paper we study the following type of problems: given an integer $k$, the problem asks for $k$ solutions such that the sum of pairwise (weighted) Hamming distances between these solutions is maximized.
    Such solutions are called \emph{diverse solutions}.
    We present a polynomial-time algorithm for finding diverse shortest $st$-paths in weighted directed graphs. 
    Moreover, we study the diverse version of other classical combinatorial problems such as diverse weighted matroid bases, diverse weighted arborescences, and diverse bipartite matchings.
    We show that these problems can be solved in polynomial time as well.
    To evaluate the practical performance of our algorithm for finding diverse shortest $st$-paths, we conduct a computational experiment with synthetic and real-world instances.
    The experiment shows that our algorithm successfully computes diverse solutions within reasonable computational time.
\end{abstract}

\section{Introduction}\label{sec:intro}

When solving a real-world problem, we usually interpret it as an instance of a mathematical optimization problem. We then seek a \emph{single} best solution that optimizes a given objective function while satisfying the set of constraints. 
Unfortunately, the \emph{single} solution obtained this way may be inadequate for practical usage, since real-life intricacies (side constraints) are oversimplified or even ignored in the optimization problem for the sake of computational feasibility.

To illustrate, consider the task of finding a fastest route between location $s$ and $t$, which may seem like an obvious instance of the shortest $st$-path problem in graphs.
However, numerous important real-life considerations -- such as traffic congestion, constructions, or availabilities of petrol stations in the route, to name a few -- cannot be factored into the graph data structure on which Dijkstra's or Bellman-Ford's algorithms are run.
Furthermore, real-life constraints are often ``...aesthetic, political, or environmental'' in nature~\cite{Baste:Diversity:2020}; such qualitative characteristics are difficult or downright impossible to express mathematically. 

In light of this issue, the \emph{multiplicity} of the solutions has gained attention. One of the best known approaches to achieve the solution multiplicity is the \emph{$k$-best enumeration}~\cite{Hamacher:K-best:1985,Eppstein:k-best:2016} (also known as \emph{top-$k$ enumeration}). The $k$-best enumeration version of an algorithm finds $k$ feasible solutions $\mathcal S = \{S_1, S_2, \ldots, S_k\}$ such that no feasible solution outside $\mathcal S$ is strictly better than that inside $\mathcal S$. This approach allows us to choose the best-suited solution amongst $k$ \emph{sufficiently good} feasible solutions, in a given real-life situation.
There are many $k$-best enumeration algorithms for various optimization problems.
See ~\cite{Hamacher:K-best:1985,Eppstein:k-best:2016} for surveys.

One potential drawback of the $k$-best enumeration approach is the lack of {\em diversity} of solutions.
Most of $k$-best enumeration algorithms, such as Lawler's framework~\cite{Lawler:procedure:1972}, recursively generate solutions from a single optimal solution $X = \{x_1, x_2, \ldots, x_t\}$ by finding a solution including $\{x_1, \ldots, x_{i-1}\}$ and excluding $x_i$ for each $1 \le i \le t$. This implies that solutions tend to be similar to each other in nature~\cite{Akgun:finding:2000}, which may potentially defeat the purpose of finding multiple solutions. 

Motivated by this, (explicitly) optimizing diversity of solutions has received considerable attention in the literature.
There are many results for finding diverse solutions for Constrained Satisfaction Problems or Mixed Integer Programming~\cite{Danna:How:2009,Hebrard:Finding:2005,Nadel:Generating:2011,Petit:Finding:2015,Petit:Enriching:2019} and diversifying query results~\cite{Drosou:Search:2010,Liu:Finding:2018,Qin:Diversifying:2012,Vieira:query:2011}.
According to \citeauthor{Baste:FPT:2019} \citeyearpar{Baste:FPT:2019}, Michael Fellows proposed \emph{the Diverse X Paradigm}, where \emph{X} is a placeholder for an optimization problem, and to investigate these problems from the theoretical perspective, particularly fixed-parameter tractability.
Based on this proposal, they initiated a theoretical study of the Diverse X Paradigm and gave several results through the lens of fixed-parameter tractability.
In particular, they showed that several diverse versions of combinatorial problems, such as {\sc Vertex Cover}, {\sc Feedback Vertex Set}, and {\sc $d$-Hitting Set}, are fixed-parameter tractable parameterized by the solution size plus the number of solutions~\cite{Baste:FPT:2019}.
Their follow-up research also discussed the fixed-parameter tractability of diverse versions of several combinatorial problems on bounded-treewidth graphs~\cite{Baste:Diversity:2020}. 

\paragraph{Diversity measures:} Before describing our results, we need to define known diversity measures and discuss known results relevant to our results.
There are mainly two diversity measures in these theoretical studies.
Let $U$ be a finite set.
Let $S_1, \ldots, S_k$ be (not necessarily disjoint) subsets of $U$.
We define
\begin{align*}
    \dist_{\rm sum}(S_1, \ldots, S_k) = \sum_{1 \le i < j \le k} |S_i \triangle S_j|,
\end{align*}
where $\triangle$ is the symmetric difference of two sets (we may call it the Hamming distance between two sets).
We also define
\begin{align*}
    \dist_{\min}(S_1, \ldots, S_k) = \min_{1 \le i < j \le k} |S_i \triangle S_j|.
\end{align*}

\paragraph{Related work:} 
There are several recent work closely related to our results.
Fomin et al.~\citeyearpar{Fomin:Diverse:2020} showed that the problem of finding two maximum matchings $M_1, M_2$ maximizing its symmetric difference $|M_1 \triangle M_2|$ in bipartite graphs can be solved in polynomial time, whereas it is NP-hard on general graphs~\cite{Fomin:Diverse:2020,Holyer:NP-completeness:1981,Hanaka:Finding:2021}.
To overcome this intractability, they devised an FPT-algorithm with respect to parameter $|M_1 \triangle M_2|$ for general graphs, that is, it runs in time $f(d) n^{O(1)}$, where $n$ is the number of vertices of an input graph and $d = |M_1 \triangle M_2|$.
Moreover, Fomin et al.~\citeyearpar{Fomin:Diverse:2021} studied the problems of finding $k$ solutions $S_1, S_2, \ldots, S_k$ for several combinatorial problems related to (linear) matroids and matchings such that the weighted version of $\dist_{\min}(S_1, S_2, \ldots, S_k)$ is at least $d$.
They showed that these problems are fixed-parameter tractable parameterized by $k + d$ (i.e., the running time of these algorithms is $f(k,d)n^{O(1)}$, where $f$ is some computable function and $n$ is the input size).
On the negative side, they showed that finding $k$ bases of a matroid maximizing the weighted version of $\dist_{\min}$ is NP-hard even on uniform matroids. 
Contrary to this hardness result, Hanaka et al.~\citeyearpar{Hanaka:Finding:2021} showed that finding $k$ bases of a matroid maximizing $\dist_{\rm sum}$ is solvable in polynomial time.\footnote{In their problem setting, the independence oracle of matroids is given as input and is assumed to be evaluated in polynomial time.}
They also proposed a general framework to obtain diverse solutions with running time exponentially depending on $k$ and the cardinality of solutions.

Finding diverse short(est) paths is a fruitful research area and a considerable amount of effort has been dedicated to it so far~\cite{Akgun:finding:2000,Chondrogiannis:Finding:2020,Liu:Finding:2018,Voss:heuristic:2015}.
Alg\"un, Erkut, and Batta \citeyearpar{Akgun:finding:2000} experimentally compared some existing methods and their method based on the dispersion problem~\cite{Kuby:dispersion:1987}.
Chondrogiannis et al. \citeyearpar{Chondrogiannis:Finding:2020} studied the problem of finding a set of $k$ paths $P_1, P_2, \ldots, P_k$ between two vertices $s$ and $t$ such that these paths are ``dissimilar'' to each other and for every path $P$ between $s$ and $t$ with $P \neq P_i$ for any $1 \le i \le k$, either there is a path $P_i$ that is similar to $P$ or the length of $P$ is not shorter than that of $P_i$ for any $1 \le i \le k$.
They showed that this problem is NP-hard and gave an exponential-time algorithm for this problem.
By \cite{Liu:Finding:2018}, a similar problem was discussed and a heuristic algorithm was given.
Voss, Moll, and Kavraki \citeyearpar{Voss:heuristic:2015} gave a heuristic algorithm that optimizes some diversity measure of paths based on the Fr\'echet distance.
To the best of our knowledge, no shortest $st$-path algorithms that run in polynomial time with a theoretical guarantee of diversity are known in the literature.

\paragraph{Our results:} In this paper, we expand the tractability border of the diverse version of classical combinatorial optimization problems.
Let $\mathcal S \subseteq 2^U$ be a set of solutions.
The common goal of our problems is to find a set of $k$ solutions $S_1, S_2, \ldots, S_k \in \mathcal S$ maximizing $\dist_w(S_1, \ldots. S_k)$, where $\dist_w$ is the weighted version of the diversity measure $d_{\rm sum}$ (see the next section for the definition).
We show that if $\mathcal S$ consists of either (1) the shortest $st$-paths in a directed graph, (2) the bases of a matroid, (3) the arborescences of a directed graphs, or (4) the bipartite matchings of size $p$, then the problem can be solved in polynomial time.
The algorithm for (2) is an extension of the algorithm of \cite{Hanaka:Finding:2021} to the weighted counterpart and that for (4) is an extension of \cite{Fomin:Diverse:2020} that allows to find more than two diverse bipartite matchings in polynomial time.
Contrary to the tractability of (1), we show that the problem of finding two disjoint ``short'' paths is NP-hard.

As will be elaborated on \Cref{lem:w-packing}, this positive outcome is achieved by understanding our problem as a generalization of packing problems.
In the packing problem (with respect to $\mathcal S$), we are asked to find \emph{mutually disjoint} sets $S_1, S_2, \ldots, S_k$ in $\mathcal S$.
Suppose that each solution in $\mathcal S$ has the same cardinality.
Hanaka et al. \citeyearpar{Hanaka:Finding:2021} showed that by appropriately duplicating and weighting each (duplicated) element in $U$, the problem of maximizing $\dist_{\rm sum}(S_1, S_2, \ldots, S_k)$ can be reduced to a certain weighted packing problem. 
Our lemma formalizes their idea and states that finding diverse solutions can be reduced to a certain packing problem as well even if solutions are weighted and have different cardinalities.
As an application of this lemma, we show that finding diverse shortest $st$-paths can be reduced to the minimum cost flow problem, which can be solved in polynomial time.

We conducted a computational experiment to measure the practical performance of our algorithm for the diverse version of the shortest path problem.
Our experiment shows that the proposed algorithm runs in reasonable time and computes more diverse shortest paths than those computed by the $k$-best enumeration algorithm. 

\section{Preliminaries}\label{sec:prel}

Let $D = (V, E)$ be a (directed) graph.
We denote by $V(D)$ and $E(D)$ the sets of vertices and edges of $D$, respectively.

Let $U$ be a finite set and let $w: U \to \mathbb Z$.
Let $S_1, \ldots, S_k$ be (not necessarily disjoint) subsets of $U$.
We define
\begin{align*}
    \dist_w(S_1, \ldots, S_k) = \sum_{1 \le i < j \le k} w(S_i \triangle S_j),
\end{align*}
where $w(X) = \sum_{x \in X} w(x)$.
This notation extends the diversity measure $\dist_{\rm sum}$ defined in the introduction, where $\dist_{\rm sum}(S_1, \ldots, S_k) = \dist_w(S_1, \ldots, S_k)$ with $w(x) = 1$ for all $x \in U$.
Let $k$ be a positive integer and let $\mathcal S \subseteq 2^U$ be a family of subsets of $U$.
We expand each element $e$ in $U$ into $k$ copies: Let $U^* = \{e_1, \ldots, e_k : e \in U\}$.
We define a function $f: U^* \to U$ such that $f(e_i) = e$ for all $e_i \in U^*$.
We say that $S^* \subseteq U^*$ is a {\em $k$-packing of $U^*$ with respect to $\mathcal S$} if $S^*$ can be partitioned into $S_1, \ldots, S_k$ such that $\{f(e^*) : e^* \in S_i\} \in \mathcal S$ for all $1 \le i \le k$.

We consider a weight function $w^*: U^* \to \mathbb Z$ such that $w^*(e_i) = w(e) \cdot (k - 2i + 1)$ for $e \in U$ and $1 \le i \le k$.
Then, we have the following lemma.

\begin{lemma}\label{lem:w-packing}
    There are $S_1, S_2, \ldots, S_k$ in $\mathcal S$ with $\dist_w(S_1, \ldots, S_k) \ge t$ if and only if there is a $k$-packing $S^*$ of $U^*$ with respect to $\mathcal S$ such that  $w^*(S^*) \ge t$.
\end{lemma}
\begin{proof}
    Suppose that there are $S_1, S_2, \ldots, S_k$ in $\mathcal S$ with $\dist_w(S_1, \ldots, S_k) \ge t$.
    For each $e \in U$, we denote by $m(e)$ the number of occurrences of $e$ in the collection $\{S_1, \ldots, S_k\}$.
    Then, we have
    \begin{align*}
        \dist_w(S_1, \ldots, S_k)
        &= \sum_{e \in U}(w(e) \cdot m(e) \cdot (k - m(e)))\\
        &= \sum_{e \in U} (w(e) \cdot \sum_{1 \le i \le m(e)} (k - 2i + 1))\\
        & = \sum_{e \in U} \sum_{1 \le i \le m(e)} w^*(e_i)\\
        & = w^*(S^*).
    \end{align*}
    
    Conversely, assume that there is a $k$-packing $S^*$ of $U^*$ with respect to $\mathcal S$ with $w^*(S^*) \ge t$.
    We assume moreover that, for each $e \in U$, $S^*$ contains consecutive elements $e_1, \ldots, e_m$ for some $m$, that is, $\{e_1, \ldots, e_m\} \subseteq S^*$ and $\{e_{m+1} \ldots, e_{k}\} \cap S^* = \emptyset$.
    This assumption is legitimate as $w(e_i) > w(e_j)$ for $1 \le i < j \le k$.
    We denote the multiplicity $m$ of $e$ by $m(e)$. 
    Let $\{S^*_1, \ldots, S^*_k\}$ be a partition of $S^*$ such that $S_i = \{f(e^*) :e^* \in S^*_i\} \in \mathcal S$ for $1 \le i \le k$.
    For each $e \in U$, the contribution of $e$ to $w^*(S^*)$ is indeed $w(e) \cdot \sum_{1 \le i \le m(e)}(k - 2i + 1)$.
    Hence, we have 
    \begin{align*}
        w^*(S^*) &= \sum_{e \in U}\sum_{1 \le i \le m(e)} w^*(e_i)\\
        &= \dist_w(S_1, \ldots, S_k)
    \end{align*}
    as in the ``only-if'' direction.
\end{proof}

Let us note that a similar idea is used in \cite{Hanaka:Finding:2021}, however, their proof relies on the fact that solutions have the same cardinality and the objective function is to maximize the sum of \emph{unweighted} Hamming distances.

For several cases, we reduce the problem of finding a largest weight $k$-packing to the minimum cost flow problem.
In the minimum cost flow problem, given a directed graph $D = (V, E)$, $s, t \in V$, a cost function $w: E \to \mathbb Q$, a capacity function $c: E \to \mathbb N_{\ge 0}$, and an integer $k$, the goal is to compute a flow $f: E \to \mathbb Q$ with flow value $k$ such that (1) $f(e) \le c(e)$ for  all $e \in E$, (2) $\sum_{(v, w) \in E}f((v,w)) = \sum_{(w, v) \in E} f((w, v))$ for all $v \in V \setminus \{s,t\}$, and (3) $\sum_{e \in E} w(e)f(e)$ is minimized, where the flow value is defined as $\sum_{(s, v) \in E} f((s, v))$.
This problem can be solved in polynomial time.
\begin{theorem}[\citeauthor{Orlin:Faster:1993}~\citeyear{Orlin:Faster:1993}]\label{thm:mcf}
    The minimum cost flow problem can be solved in time $O(|E|^2 \log^2 |V|)$.
\end{theorem}

\section{Diverse Shortest $st$-Paths}\label{sec:path}

As an application of \Cref{lem:w-packing}, we consider the following problem.

\begin{definition}
    Let $D = (V, E)$ be a directed graph with specified vertices $s, t \in V$.
    Let $\ell : E \to \mathbb N_{> 0}$ be a length function on edges.
    Let $\mathcal P$ be the set of all shortest paths from $s$ to $t$ in $(D, \ell)$.
    Given an integer $k$ and a weight function $w: E \to \mathbb Z$, {\sc Diverse Shortest $st$-Paths} asks for $k$ paths $P_1, \ldots, P_k \in \mathcal P$ such that $\dist_w(E(P_1), \ldots, E(P_k))$ is maximized.
\end{definition}

\begin{theorem}\label{thm:dsp}
    {\sc Diverse Shortest $st$-Paths} can be solved in $O(k^2|E|^2\log^2 |V|)$ time.
\end{theorem}

We first compute the shortest distance label ${\rm dist}: V \to \mathbb N_{\ge 0}$ from $s$ in time $O(|E| + |V| \log |V|)$ by Dijkstra's single source shortest path algorithm.
For each edge $e = (u, v) \in E$ with ${\rm dist}(u) \neq {\rm dist}(v) - \ell(e)$, we remove it from $D$.
We also remove vertices that are not reachable from $s$ and not reachable to $t$ in the removed graph.
This does not change any optimal solution since every path in $\mathcal P$ does not include such vertices and edges.
Then, the obtained graph, denoted by $D' = (V', E')$, has no directed cycles.

\begin{observation}\label{obs:sp}
    Every path in $\mathcal P$ is a path from $s$ to $t$ in $D'$.
    Moreover, every path from $s$ to $t$ in $D'$ belongs to $\mathcal P$.
\end{observation}

From this directed acyclic graph $D'$, we construct a weighted directed multigraph $D^*$ by replacing each edge $e = (u, v)$ with $k$ copies $e_1, \ldots, e_k$ and setting $w^*(e_i)$ to $w(e)\cdot(k - 2i + 1)$ for $e \in E'$ and $1 \le i \le k$.
By~\Cref{lem:w-packing}, it is sufficient to find a maximum weight $k$-packing of $E(D^*)$ with respect to $\mathcal P$.

\begin{lemma}
    Let $D^*$ and $w: E(D^*) \to \mathbb N_{\ge 0}$ be as above. Then, there is a polynomial-time algorithm that finds a maximum weight $k$-packing of $E(D^*)$ with respect to $\mathcal P$.
\end{lemma}
\begin{proof}
    We reduce the $k$-packing problem to the minimum-cost flow problem, which can be solved in polynomial time (by \Cref{thm:mcf}).
    The source and the sink vertices of $D^*$ are defined as $s$ and $t$, respectively. For each $e \in E(D^*)$, we assign the capacity value of 1 (to prevent edge sharing) and the cost value of $-w^*(e)$. The flow requirement is set to $k$.
    Then we can find a flow $f: E^* \to \mathbb R_{\ge 0}$ maximizing $\sum_{e \in E^*}f(e) \cdot w^*(e)$ in polynomial time.
    Moreover, it is well known that $f$ is integral, that is, $f(e) \in \mathbb N_{\ge 0}$ for every $e \in E^*$, as the capacity is integral.
    Since all of the edges in $E^*$ has a capacity of 1, $f$ can be decomposed into $k$ edge-disjoint $st$-paths $P_1, P_2,  \ldots, P_k$, which implies that the maximum weight $k$-packing with respect to $\mathcal P$ can be found in polynomial time as well.
\end{proof}

Since $D^*$ has at most $|V|$ vertices and at most $k|E|$ edges, \Cref{thm:dsp} follows.

We note that our algorithm can be used for finding diverse shortest $st$-paths in \emph{undirected} graphs.
For undirected graphs, we first compute the shortest distance label ${\rm dist}$.
Then, for each (undirected) edge $e = \{u, v\}$, we orient $e$ as $(u, v)$ if ${\rm dist}(u) = {\rm dist}(v) - \ell(e)$ and remove $e$ otherwise.
Then, the obtained graph $D'$ is indeed a directed acyclic graph and satisfies \Cref{obs:sp}.
Hence we can apply the same algorithm to $D'$.

In practical situations, it suffices to find diverse ``nearly shortest'' paths.
More specifically, given the same instance of {\sc Diverse Shortest $st$-Paths} and a positive threshold $\theta$, we are asked to find a set of $k$ paths $P_1, P_2, \ldots, P_k$ from $s$ to $t$ in $(D, \ell)$ such that the length of $P_i$ is at most $\theta$ and $\dist_w(E(P_1), E(P_2), \ldots, E(P_k))$ is maximized.
Here, the length of a path $P_i$ is defined as $\sum_{e \in E(P_i)}\ell(e)$.
Let us note that paths are simple, that is, every vertex appears at most once in the path.
The problem, which we call {\sc Diverse Short $st$-Paths}, indeed generalizes {\sc Diverse Shortest $st$-Paths}: If $\theta$ is the shortest path distance from $s$ to $t$ in $D$, the problem corresponds to {\sc Diverse Shortest $st$-Paths}.
Unfortunately, this generalization is intractable.

\begin{theorem}
    {\sc Diverse Short $st$-Paths} is NP-hard even if $k = 2$ and $w(e) = \ell(e) = 1$ for all $e \in E$.
\end{theorem}
\begin{proof}
    We perform a polynomial-time reduction from the $st$-Hamiltonian path problem, which is well known to be NP-complete~\cite{Karp:Complexity:1972}.
    In this problem, we are given a directed graph $D = (V, E)$ and $s, t \in V$ with $s \neq t$ and asked to determine whether $D$ has a Hamiltonian path from $s$ to $t$.
    From $D$, we construct a directed graph $D'$ by adding a path $P$ from $s$ to $t$ of length $|V|$.
    Now, we claim that $D$ has a Hamiltonian path from $s$ to $t$ if and only if $D'$ has two $st$-paths $P_1$ and $P_2$ of length at most $\theta = |V|$ with $|E(P_1) \triangle E(P_2)| \ge 2|V| - 1$.
    
    Suppose that $D$ has a Hamiltonian path $P'$ from $s$ to $t$.
    Then, $P'$ is disjoint from $P$ and hence $|E(P') \triangle E(P)| = |E(P')| + |E(P)| = 2|V| - 1$.
    
    Conversely, assume that $D'$ has two $st$-paths $P_1$ and $P_2$ of length at most $|V|$ with $|E(P_1) \triangle E(P_2)| \ge 2|V| - 1$.
    Since every $st$-path $P'$ distinct from $P$ has length at most $|V| - 1$, one of the two paths, say $P_1$, must be $P$.
    Moreover, such a path $P'$ is disjoint from $P$ as otherwise $P'$ must contain $s$ or $t$ as an internal vertex.
    This implies that $P_2$ is also an $st$-path in $D$.
    As $|E(P_1) \triangle E(P_2)| \ge 2|V| - 1$, the length of $P_2$ is at least $|V| - 1$, implying that $P_2$ is a Hamiltonian path from $s$ to $t$ in $D$.
\end{proof}

\section{Diverse Matroid Bases and Arborescences}
In this section, we discuss some tractable problems related to matroids.
Let $E$ be a finite set.
A \emph{matroid} $\mathcal M$ on $E$ is a pair $(E, \mathcal I)$ with $\mathcal I \subseteq 2^E$ satisfying the following axioms:
(1) $\emptyset \in \mathcal I$; 
(2) For $X \in \mathcal I$, every $Y \subseteq X$ is contained in $\mathcal I$;
(3) For $X, Y \in \mathcal I$ with $|X| < |Y|$, there is $e \in Y$ such that $X \cup \{e\} \in \mathcal I$.
Every set in $\mathcal I$ is called an \emph{independent set} of a matroid $\mathcal M = (E, \mathcal I)$.
A \emph{base} of $\mathcal M$ is an inclusion-wise maximal independent set of $\mathcal M$.
The first part of this section is devoted to develop a polynomial-time algorithm for the following problem.

\begin{definition}
    Given a matroid $\mathcal M = (E, \mathcal I)$ with a weight function $w: E \to \mathbb Z$ and an integer $k$, {\sc Weighted Diverse Matroid Bases} asks for $k$ bases $B_1, \ldots, B_k$ of $\mathcal M$ such that $\dist_w(B_1, \ldots, B_k)$ is maximized.
\end{definition}

In \cite{Hanaka:Finding:2021}, they consider a special case of {\sc Weighted Diverse Matroid Bases} where each element in the ground set $E$ has a unit weight and give a polynomial-time algorithm for it, assuming that the independence oracle $\mathcal I$ can be evaluated in polynomial time.
This result is obtained by reducing the problem to that of finding disjoint bases of a matroid, which can be solved in polynomial time.

\begin{theorem}[\cite{Edmonds:Matroid:1968,Nash-Williams:application:1967}]\label{thm:matroid-union}
    Let $\mathcal M = (E, \mathcal I)$ be a matroid and let $w: E \to \mathbb Z$.
    Suppose that the membership of $\mathcal I$ can be checked in polynomial time.
    Then, the problem of deciding whether there is a set of mutually disjoint $k$ bases $B_1, \ldots, B_k$ of $\mathcal M$ can be solved in polynomial time.
    Moreover, if the answer is affirmative, we can find such bases that maximize the total weight (i.e., $\sum_{1 \le i \le k} \sum_{e \in B_k} w(e)$) in polynomial time.
\end{theorem}

By applying \Cref{lem:w-packing}, we have a polynomial-time algorithm for {\sc Weighted Diverse Matroid Bases} as well.

\begin{theorem}\label{thm:base}
    {\sc Weighted Diverse Matroid Bases} can be solved in polynomial time.
\end{theorem}

\begin{proof}
    The proof is almost analogous to that in \cite{Hanaka:Finding:2021}.
    Let $\mathcal M = (E, \mathcal I)$ be a matroid and let $e \in E$.
    Define $\mathcal J = \mathcal I \cup \{(F \setminus \{e\}) \cup \{e'\} : F \in \mathcal I \land e \in F\}$.
    Then, $(E \cup \{e'\}, \mathcal J)$ is also a matroid~\cite{Hanaka:Finding:2021,Nagamochi:Complexity:1997}.
    We define $k$ copies $e_1, e_2, \ldots, e_k$ for each $e \in E$ and $E^* = \{e_1, \ldots, e_k : e \in E\}$.
    Then, the pair $\mathcal M^* = (E^*, \mathcal I^*)$ is a matroid if $\mathcal I^*$ consists of all sets $F \subseteq 2^{E^*}$ such that $F$ contains at most one copy of $e_1, \ldots, e_k$ for each $e \in E$ and $\bigcup_{e_i \in F} f(e_i) \in \mathcal I$, where $f(e_i) = e$ for $e \in E$ and $1 \le i \le k$.
    
    To find a set of $k$ bases $B_1, \ldots, B_k$ of $\mathcal M$ maximizing $\dist_w(B_1, \ldots, B_k)$, by \Cref{lem:w-packing}, it suffices to find a maximum weight $k$-packing with respect to the base family of $\mathcal M^*$ under a weight function $w^*$ with $w^*(e_i) = w(e)\cdot(k - 2i + 1)$ for $e \in E$ and $1 \le i \le k$, which can be solved in polynomial time by \Cref{thm:matroid-union}.
\end{proof}

\Cref{thm:base} allows us to find diverse spanning trees in undirected graphs as the set of edges in a spanning tree forms a base of a graphic matroid~\cite{Oxley:Matroid:2006}.
In this section, we develop a polynomial-time algorithm for a directed version of this problem.
Let $D = (V, E)$ be a directed graph and let $r \in V$.
We say that a subgraph $T$ of $D$ is an {\em arborescence} (with root $r$) if for every $v \in V$, there is exactly one directed path from $r$ to $v$ in $T$.
In other words, an arborescence is a spanning subgraph of $D$ in which each vertex except $r$ has in-degree one and its underlying undirected graph is a tree.
The problem we consider here is defined as follows.

\begin{definition}
    Given an edge-weighted directed graph $D = (V, E)$ with weight function $w: E \to \mathbb Z$, $r \in V$, and an integer $k$, {\sc Weighted Diverse Arborescences} asks for $k$ arborescences $T_1, \ldots, T_k$ of $D$ with root $r$ such that $\dist_w(E(T_1), \ldots, E(T_k))$ is maximized.
\end{definition}

Note that we cannot apply \Cref{thm:base} to this problem as we do not know whether the set of arborescences can be defined as the family of bases of a matroid.
However, Edmonds \citeyearpar{Edmonds:Some:1975} gave a polynomial-time algorithm for finding $k$ arc-disjoint arborescences with the maximum total weight.

\begin{theorem}\label{thm:arb}
    {\sc Weighted Diverse Arborescences} can be solved in polynomial time.
\end{theorem}

The proof of \Cref{thm:arb} is almost analogous to that in \Cref{thm:base}.
We define a directed graph $D^*$ with vertex set $V$ from $D = (V, E)$ such that for $e = (u, v) \in E$, we add $k$ parallel edges $e_1, \ldots, e_k$ directed from $u$ to $v$ to $D'$.
Then, we set $w^*(e_i) = w(e)\cdot(k - 2i + 1)$ for $e \in E$ and $1 \le i \le k$.
By \Cref{lem:w-packing}, it is sufficient to find a maximum weight $k$-packing of $E(D^*)$ with respect to the family of arborescences of $D^*$, which can be found in polynomial time by the following result.

\begin{theorem}[\citeauthor{Edmonds:Some:1975}~\citeyear{Edmonds:Some:1975}]\label{thm:arb-packing}
    Given an edge-weighted directed multigraph $D = (V, E)$ with weight function $w: E \to \mathbb Z$, $r \in V$, and an integer $k$, the problem of finding $k$ edge-disjoint arborescences $T_1, \ldots, T_k$ with root $r$ maximizing $w(\bigcup_{1 \le i \le k} E(T_i))$ is solved in polynomial time.  
\end{theorem}

\section{Diverse Bipartite Matchings}\label{sec:bm}
A {\em matching} of a graph $G = (V, E)$ is a set $M \subseteq E$ of edges such that no two edges share their end vertices.
In this section, we consider the following problem.
\begin{definition}
    Let $G = (A \cup B, E, w)$ be an edge-weighted bipartite graph with $w: E \to \mathbb Z$, where $A$ and $B$ are the color classes of $G$.
    Let $k, p$ be positive integers.
    We denote by $\mathcal M$ the collection of all matchings of $G$ with cardinality exactly $p$.
    {\sc Diverse Bipartite Matchings} asks for $k$ matchings $M_1, \ldots, M_k \in \mathcal M$ maximizing $\dist_w(M_1, \ldots, M_k)$.
\end{definition}
As mentioned in the introduction, the problem of finding two edge-disjoint perfect matchings in a general graph is known to be NP-complete~\cite{Holyer:NP-completeness:1981,Fomin:Diverse:2020,Hanaka:Finding:2021}, while the problem of finding two maximum matchings $M_1$ and $M_2$ maximizing $|M_1 \triangle M_2|$ can be solved in polynomial time on bipartite graphs~\cite{Fomin:Diverse:2020}.
In this section, we design a polynomial-time algorithm for {\sc Diverse Bipartite Matchings} by applying \Cref{lem:w-packing}, which generalizes the aforementioned polynomial-time algorithm of \cite{Fomin:Diverse:2020}.

We construct a bipartite multigraph $G^*$ from $G$ by replacing each edge $e = \{a, b\} \in E$ with $k$ parallel edges $e_1, \ldots, e_k$.
We set $w^*(e_i) = w(e) \cdot (k - 2i + 1)$ for each $e \in E$ and $1 \le i \le k$.
By~\Cref{lem:w-packing}, it suffices to show that there is a polynomial-time algorithm that computes a maximum weight $k$-packing of $E^*$ with respect to $\mathcal M$, where $\mathcal M$ is the collection of matchings $M$ of $G$ with cardinality exactly $p$.
This problem can be solved in polynomial time by reducing to the minimum cost flow problem as follows.

Let $G^* = (A \cup B, E^*)$ be bipartite and let $w^*: E^* \to \mathbb Z$.
We construct a directed acyclic graph from $G^*$ by orienting each edge $\{a, b\}$ of $G^*$ from $a$ to $b$, where $a \in A$ and $b \in B$.
Each arc $(a, b)$ in $G^*$ has capacity one and cost $-w^*(\{a, b\})$.
We also add a source vertex $s$, a sink vertex $t$, and then arcs $(s, a)$ for $a \in A$ and $(b, t)$ for $b \in B$.
The arcs incident to the source or the sink have capacity $k$ and have cost zero.
Now, we set the flow requirement from $s$ to $t$ to $kp$.
Thanks to the integral theorem of the minimum cost flow problem, by \Cref{thm:mcf}, we can find in polynomial time a maximum weight subgraph $H^*$ of $G^*$ such that $H^*$ has exactly $kp$ edges and each vertex has degree at most $k$.
From this subgraph $H^*$, we need to construct a maximum weight $k$-packing of $E^*$ with respect to $\mathcal M$.
The following lemma ensures that it is always possible.

\begin{lemma}
    Let $H^*$ be a bipartite graph with $kp$ edges.
    Suppose the maximum degree of a vertex in $H^*$ is at most $k$.
    Then, the edges of $H^*$ can be partitioned into $k$ matchings of cardinality exactly $p$.
    Moreover, such a partition can be computed in polynomial time from $H^*$.
\end{lemma}
\begin{proof}
    Every bipartite graph of maximum degree at most $k$ has a proper edge-coloring with $k$ colors, and such an edge coloring can be computed in polynomial time~\cite{Cole:Edge:1982,Gabow:Algorithms:1982}.
    We can assume that each color is used at least once by recoloring an edge whose color is used at least twice.
    Then, the edges of $H^*$ can be decomposed into $k$ non-empty matchings $M_1, \ldots, M_k$.
    If $|M_1| = \cdots = |M_k| = p$, we are done.
    Suppose that there is a pair of matchings $M_i$ and $M_j$ with $|M_i| > p$ and $|M_j| < p$.
    Then $M_i \triangle M_j$ is a subgraph of $H^*$ of maximum degree at most two.
    As $|M_i| > |M_j|$, the subgraph contains an augmenting path $P = (v_1, \ldots, v_t)$ with $\{v_{\ell}, v_{\ell + 1}\} \in M_i$ for odd $\ell$ and $\{v_{\ell}, v_{\ell + 1}\} \in M_j$ for even $\ell$.
    Let $M'_i = (M_i \setminus E(P)) \cup (E(P) \cap M_j)$ and $M'_j = (M_j \setminus E(P)) \cup (E(P) \cap M_i)$.
    Then, we have matchings $M'_i$ and $M'_j$ such that $|M'_i| = |M_i| - 1$, $|M'_j| = |M_j| + 1$, and $M_i \cup M_j = M'_i \cup M'_j$.
    By repeating this argument, we have a desired set of matchings.
\end{proof}

Therefore, $E(H^*) = M_1 \cup \cdots \cup M_k$ is a maximum weight $k$-packing of $E(H^*)$ with respect to $\mathcal M$, which implies the following theorem.

\begin{theorem}\label{thm:dbm}
    {\sc Diverse Bipartite Maximum Matching} can be solved in polynomial time.
\end{theorem}

\section{Computational Experiments}
In order to assess the practical performance of our algorithm for {\sc Diverse Shortest $st$-Paths}, a computational experiment was conducted on a computer equipped with Intel(R) Xeon(R) Gold 5122 processor (3.60GHz) and 92GB RAM. Our diverse shortest $st$-paths algorithm\footnote{The code is available at \url{https://github.com/Dotolation/diverse-graph-algo}.} was implemented using Java with JGraphT library~\cite{jgrapht}\footnote{\url{https://jgrapht.org/}}. We compared our algorithm with Eppstein's $k$-shortest paths algorithm~\citeyearpar{Eppstein:Finding:1998}, which is included in JGraphT by default. 

\paragraph{Datasets:}
We used one type of synthetic graphs (Grid) and two types of real-world graphs (SNAP and DIMACS).

\begin{itemize}
    \item {\bf Grid}: We generated $p \times p$-grid graphs with vertex set $V = \{1, 2, \ldots, p\} \times \{1, 2, \ldots, p\}$ and edge set $E = \{\{(i, j), (i', j')\} : |i-i'| + |j-j'| = 1\}$ for $p \in \{40, 50, \ldots, 140\}$.
    We set $s$ and $t$ to the ``top-left'' corner vertex $(1, 1)$ and the ``bottom-right'' corner vertex $(p, p)$, respectively.
    Each edge has a unit weight and hence every edge belongs to a shortest $st$-path.
    
    
    \item {\bf SNAP}: We selected three directed unweighted graphs from Stanford Large Network Dataset Collection~\footnote{\url{http://snap.stanford.edu/index.html}}: wiki-Vote ($|V|=7,115$, $|E|=103,689$), soc-Slashdot0922 ($|V|=82,168$, $|E|=870,161$), and ego-Twitter ($|V|=81306$, $|E|=1,768,149$).
    These graphs are directed and unweighted. 
    As the diameters (i.e., the maximum distance between any pair of vertices) of these graphs are relatively small, the shortest path distance between $s$ and $t$ is indeed small (see {\bf Results}).
    
    \item {\bf DIMACS}: We selected two directed edge-weighted graphs from the 9th DIMACS Implementation Challenge -- Shortest Paths\footnote{\url{https://www.diag.uniroma1.it/challenge9/}}: NY ($|V| = 264,346$, $|E| = 733,846$) and FLA ($|V| = 1,070,376$, $|E| = 2,712,798$). Each edge weight of DIMACS graphs were rounded to the nearest 100 (any value below 100 was set to 100) for rudimentary data smoothing, which may increase the number of shortest $st$-paths.
\end{itemize}

\paragraph{Method:}
In principle, the source-sink vertex pair used in each test instance was selected randomly (seed $=2021$). However, given an $st$-pair, if the number of the shortest $st$ paths were less than $3k$, the pair was excluded from our experiment, since the point of our discussion is to evaluate the ability to pick $k$ solutions in a diverse manner from the much larger solution set. Likewise, we also excluded $st$ pairs whose average shortest path length (unweighted) falls short of 3 as it is difficult to assess the diversity of such short paths. 
Using the randomly selected $n$ vertex pairs with $n = 400$ for each instance, our algorithm and the $k$-best algorithm were executed.
We evaluate \emph{processing time} (in milliseconds) and the \emph{diversity measure} (the pair-wise Hamming distance of $k$ shortest paths) for these algorithms.

A different method is used when testing Grid graphs. Each Grid graph created with $p \in \{40,50,\ldots,140\}$ was tested once using the source vertex $(1,1)$ and the sink vertex $(p,p)$. 
The processing time and diversity were gauged at three distinct values of $k$ ($k \in \{10,50,100\}$).  

\paragraph{Results:}
For each real-world graph, $n = 400$ random $st$-pairs were sampled to test both algorithms; the number $k$ of solutions was fixed to $10$.  

\begin{table}[h!]
\centering
\caption{Average $st$-path length and count ($n = 400$, $k=10$).}\label{tbl:length}
 \begin{tabular}{||c||r|r||} 
 \hline
 Name & Length & Count\\ [0.5ex] 
 \hline\hline
 SNAP (Wiki-Vote) & 4.28 & 58.85 \\ 
 \hline
 SNAP (soc-Slashdot0922) & 4.59 & 100.14 \\
 \hline
 SNAP (ego-Twitter) & 5.40 & 136.56 \\
 \hline
 \hline
 DIMACS (NY) & 523.41 & 1307.32 \\
 \hline
 DIMACS (FLA) & 1496.00 &  2205.34\\ 
 \hline
 \end{tabular}
\end{table}

\begin{table}[h!]
\centering
\caption{Average processing time (ms) ($n=400$, $k=10$).}\label{tbl:runtime}
 \begin{tabular}{||c||r|r||} 
 \hline
 Name & Ours & $k$-Best\\ [0.5ex] 
 \hline\hline
 SNAP (Wiki-Vote) & 20.65 & 17.24 \\ 
 \hline
 SNAP (soc-Slashdot0922) & 603.95 & 583.30 \\
 \hline
 SNAP (ego-Twitter) & 1027.89 & 1021.89 \\
 \hline
 \hline
 DIMACS (NY) & 1244.96 & 1223.76 \\
 \hline
 DIMACS (FLA) & 6319.66  & 6287.20 \\ 
 \hline
 \end{tabular}
\end{table}

The SNAP dataset consists of social networks characterized by high vertex-to-edge ratios and low diameters. Unsurprisingly, their average (unweighted) shortest $st$-path length was significantly low as \Cref{tbl:length} demonstrates. Note that our experiment excluded the $st$-pairs whose shortest path distance is below 3.
This is contrasted with DIMACS graphs having very high average path length, which is at least 100 times greater than those of SNAP. DIMACS graphs also have significantly higher shortest path counts (i.e., the number of shortest $st$-paths) as well.  

As stated on \Cref{tbl:runtime}, the $k$-best enumeration approach was marginally superior to our diverse approach in terms of the processing time.
However, the difference might be small enough to be ignored in practical situations, and this observation held true for both SNAP and DIMACS datasets. 
Let us note that the theoretical running time bounds are quite different: our algorithm runs in time $O(k^2|E|^2\log^2 |V|)$ (\Cref{thm:dsp}), while the $k$-shortest paths algorithm runs in time $O(|E| + |V|(k + \log |V|))$ \cite{Eppstein:Finding:1998}.
Regardless the difference between the computational time, both algorithms seems to run reasonably fast on moderately large real-life graphs, such as DIMACS (FLA) that has over a million vertices and two million edges.  

\begin{table}[h!]
\centering
\caption{Average diversity of solutions ($n=400$, $k=10$).}\label{tbl:diversity}
 \begin{tabular}{||c||r|r||} 
 \hline
 Name & Ours & $k$-Best\\ [0.1ex] 
 \hline\hline
 Wiki-Vote & \num{3.198e+2} & \num{2.891e+2}\\ 
 \hline
 soc-Slashdot0922 & \num{3.133e+2} & \num{2.791e+2}\\
 \hline
 ego-Twitter & \num{4.150e+2}& \num{3.415e+2}\\
 \hline
 \hline
 NY & \num{2.280e+6} & \num{1.059e+6}\\
 \hline
 FLA & \num{3.610e+6} & \num{1.406e+6}\\ 
 \hline
 \end{tabular}
\end{table}

On the other hand, our algorithm consistently demonstrated greater diversity than the $k$-best algorithm (\Cref{tbl:diversity}). 
The difference in the diversity measure was especially pronounced in the DIMACS dataset (in particular DIMACS (FLA)), which has noticeably longer average shortest path length. In case of the SNAP graphs, the difference in the diversity measure was observable, but significantly less dramatic. In essence, the effectiveness of our algorithm seems to correlate with the average length and the total number of shortest $st$-paths, which is in line with our prediction.

\begin{figure*}[h]
\begin{minipage}{.5\textwidth}
\begin{center}
\begin{tikzpicture}[scale=0.75]
\begin{axis}[
    xlabel={Grid dimension ($p$)},
    ylabel={Processing time (ms)},
    xmin=40, xmax=140,
    ymin=10, ymax=1000000000,
    ymode=log,
    xtick={60,80,100,120,140},
    ytick={10,100,1000,10000,100000,1000000,10000000,100000000, 1000000000, 1000000000},
    legend pos=north west,
    legend columns=2,
    ymajorgrids=true,
    grid style=dashed,
]

\addplot[
    color=blue,
    mark=square,
    ]
    coordinates {
    (40,417)(50,855)(60,1711)(70,2672)(80,3974)(90, 6914)(100,10078)(110,15426)(120,28269)(130,28269)(140,35659)
    };
    
\addplot[
    color=orange,
    mark=triangle,
    ]
    coordinates {
    (40,13)(50,22)(60,187)(70,329)(80,283)(90, 665)(100,8993)(110,14014)(120,6422)(130,96249)(140,159480)
    };
    
\addplot[
    color=purple,
    mark=square,
    ]
    coordinates {
    (40,5747)(50,14454)(60,27981)(70,44892)(80,69488)(90, 113553)(100,160743)(110,261756)(120,327085)(130,384838)(140,429891)
    };

\addplot[
    color=green,
    mark=triangle,
    ]
    coordinates {
    (40,13)(50,21)(60,186)(70,342)(80,288)(90, 639)(100,8965)(110,13999)(120,6360)(130,95845)(140,159496)
    };
    
\addplot[
    color=black,
    mark=square,
    ]
    coordinates {
    (40,24154)(50,52317)(60,108028)(70,185222)(80,254522)(90,382724)(100,624933)(110,751876)(120,968132)(130,1240217)(140,1583424)
    };
    
\addplot[
    color=brown,
    mark=triangle,
    ]
    coordinates {
    (40,13)(50,21)(60,173)(70,332)(80,287)(90, 663)(100,8982)(110,13939)(120,6529)(130,91914)(140,162124)
    };

\legend{Ours ($k=10$)), $k$-best ($k=10$), Ours ($k=50$), $k$-best ($k=50$), Ours ($k=100$), $k$-best ($k=100$)}
\end{axis}
\end{tikzpicture}
\caption{Grid dimension vs. processing time}\label{fig:runtimeG}
\end{center}
\end{minipage}
\hfill
\begin{minipage}{.5\textwidth}
\begin{center}
\begin{tikzpicture}[scale=0.75]
\begin{axis}[
    xlabel={Grid dimension ($p$)},
    ylabel={diversity measure},
    xmin=40, xmax=140,
    ymin=100, ymax=1000000000,
    ymode=log,
    xtick={60,80,100,120,140},
    ytick={10,100,1000,10000,100000,1000000,10000000,100000000,1000000000,1000000000},
    legend pos=north west,
    legend columns=2,
    ymajorgrids=true,
    grid style=dashed,
]

\addplot[
    color=blue,
    mark=square,
    ]
    coordinates {
    (40,6876)(50,8676)(60,10476)(70,12276)(80,14076)(90, 15876)(100,17676)(110,19476)(120,21276)(130,23076)(140,24876)
    };
    
\addplot[
    color=orange,
    mark=triangle,
    ]
    coordinates {
    (40,542)(50,686)(60,662)(70,686)(80,662)(90, 430)(100,686)(110,662)(120,662)(130,662)(140,430)
    };
    
\addplot[
    color=purple,
    mark=square,
    ]
    coordinates {
    (40,182652)(50,231652)(60,280652)(70,329652)(80,378652)(90, 427652)(100,476652)(110,525652)(120,574652)(130,623652)(140,672652)
    };

\addplot[
    color=green,
    mark=triangle,
    ]
    coordinates {
    (40,18084)(50,20376)(60,19520)(70,20376)(80,19520)(90, 18546)(100,20376)(110,19520)(120,19520)(130,19520)(140,18546)
    };
    
\addplot[
    color=black,
    mark=square,
    ]
    coordinates {
    (40,731832)(50,928688)(60,1126008)(70,1323792)(80,1521792)(90,1719792)(100,1917792)(110,2115792)(120,2313792)(130,2511792)(140,2709792)
    };
    
\addplot[
    color=brown,
    mark=triangle,
    ]
    coordinates {
    (40,81838)(50,78882)(60,86686)(70,78882)(80,86686)(90,83092)(100,78882)(110,86686)(120,86686)(130,86686)(140,83092)
    };

\legend{Ours ($k=10$)), $k$-best ($k=10$), Ours ($k=50$), $k$-best ($k=50$), Ours ($k=100$), $k$-best ($k=100$)}
\end{axis}
\end{tikzpicture}
\caption{Grid dimension vs. diversity measure}\label{fig:diversityG}
\end{center}
\end{minipage}
\end{figure*}

\Cref{fig:runtimeG} 
shows the running time of our algorithm and the $k$-best algorithm on Grid graphs with varying values of $k$ across grid dimensions $p$.\footnote{It might be hard to distinguish the lines for the $k$-best algorithm in \Cref{fig:runtimeG} as their running times are almost the same for $k = 10, 50, 100$.}
Our diverse algorithm was considerably slower than the $k$-best algorithm, regardless the value of $k$.
The growth of $k$ resulted in a significant slowdown of the diverse algorithm, whereas it had negligible impact on the processing time of the $k$-best algorithm.
As every single edge in a grid graph is included in a shortest path from $s$ to $t$, no edge or vertex is discarded in the preprocessing step. (This phenomenon is unlikely in most real-world graphs.) This increases the computational burden of our algorithm on the minimum cost flow step, which explains the experimental outcome. 

However, our algorithm vastly outperformed the $k$-best approach in regard to diversity maximization for all values of $k$ (\Cref{fig:diversityG}). As the diversity is measured by using edge-wise Hamming distance, the growth of $k$ has a strong impact on the diversity measure by the virtue of adding more edges to compare. Despite this, the diverse algorithm with $k=50$ has outperformed the $k$-best counterpart with $k=100$ for every grid dimension with $p \ge 40$.

In summary, our experiment suggests that our diverse algorithm gives significantly more diverse solutions compared to the conventional $k$-best enumeration algorithm. Although its slow speed on synthetic grid graphs may raise a question of practicality, it should be emphasized that our algorithm is efficient on large real-life data with over 2 million edges, with the processing time on par with the JGraphT implementation of $k$-best algorithm.

\section{Concluding Remarks}
In this paper, we study the problem of finding $k$ shortest $st$-paths $P_1, P_2, \ldots, P_k$ in edge-weighted directed graphs maximizing $\dist_w(E(P_1), E(P_2), \ldots, E(P_k))$, namely {\sc Diverse Shortest $st$-Paths}.
We show that this problem is polynomial-time solvable, while a slightly ``relaxed'' version {\sc Diverse Short $st$-Paths} is NP-hard.
The positive result is shown by reducing to a maximum weight packing problem, which also proves the tractability of several problems, such as {\sc Weighted Diverse Matroid Bases}, {\sc Weighted Diverse Arborescences}, and {\sc Diverse Bipartite Matchings}.
Our result of {\sc Diverse Bipartite Matchings} extends the result of \cite{Fomin:Diverse:2020}, in which they gave a polynomial-time algorithm for finding two maximum bipartite matchings $M_1, M_2$ maximizing $\dist(M_1, M_2)$, to general $k \ge 2$.
We also conducted an experiment to assess the practical performance of our algorithm for {\sc Diverse Shortest $st$-Paths}.
Our experiment shows that the practical running time of our algorithm is comparable with that of the known $k$-best enumeration algorithm~\cite{Eppstein:Finding:1998} for real-world graph.
For synthetic graphs, we used grid graphs, in which every edge belongs to a shortest path between the top-left corner and the bottom-right corner.
Our algorithm works slower than the $k$-best enumeration algorithm on grid graphs.
However, the running time of our algorithm is still reasonable for moderate-sized instances, while the diversity of paths is significantly larger than that computed by the $k$-best enumeration algorithm.

There are several open questions related to our results.
We gave a polynomial-time algorithm for {\sc Diverse Shortest $st$-Paths}.
However, it might be more acceptable in some applications to use $\dist_{\min}$ as its diversity measure.
It would be interesting to know the computational complexity of the problem of finding $k$ shortest $st$-paths $P_1, P_2, \ldots, P_k$ maximizing $\dist_{\min}(E(P_1), E(P_2), \ldots, E(P_k))$.
Contrary to this tractability, we show that {\sc Diverse Short $st$-Paths} is NP-hard.
In our proof, we use a relatively large threshold $\theta = |V|$ of the length of $st$-paths.
Another interesting question is whether {\sc Diverse Short $st$-Paths} is still hard even if $\theta$ is close to the shortest $st$-path distance, which might be a reasonable restriction in practice.
{\sc Weighted Diverse Arborescences} and {\sc Diverse Bipartite Matchings} are special cases of the following problem: Given two matroids $\mathcal M_1 = (E, \mathcal I_1)$ and $\mathcal M_2 = (E, \mathcal I_2)$ of the same ground set $E$ and an integer $k$, we are asked to find $k$ common bases $B_1, B_2, \ldots, B_k$ of $\mathcal M_1$ and $\mathcal M_2$ maximizing $\dist_w(B_1, B_2, \ldots, B_k)$.
Here, a subset $B \subseteq E$ is a common base of $\mathcal M_1$ and $\mathcal M_2$ if $B$ is a base of both matroids.
This problem is at least as hard as the problem of packing $k$ common bases of two matroids, implying that this problem cannot be solved with a polynomial number of independence oracle queries~\cite{Berczi:Complexity:2021}.
This raises the following question: Is there another natural problem that is a special case of finding diverse common bases and solvable in polynomial time?

\section*{Acknowledgments} 
The authors thank anonymous referees for valuable comments.
This work is partially supported by JSPS Kakenhi Grant Numbers
JP18H04091, 
JP18K11168, 
JP18K11169, 
JP19H01133, 
JP20H05793, 
JP20K19742, 
JP21H03499, 
JP21K11752, 
JP21K17707 
JP21H05852 
JP21K17812, 
and JST CREST Grant Number JPMJCR18K3. 

\bibliography{ref}

\end{document}